\documentclass[11pt]{article}
\usepackage{fullpage,times,amsmath,amssymb,amsthm,hyperref,graphicx,xcolor}
\usepackage[T1]{fontenc}

\newtheorem{theorem}{Theorem}[section]
\newtheorem{lemma}[theorem]{Lemma}
\newtheorem{corollary}[theorem]{Corollary}
\newtheorem{problem}[theorem]{Problem}
 
\newcommand{\eps}{\varepsilon}
\newcommand{\OO}{\widetilde{O}}
\newcommand{\Ex}{\mathbb{E}}
\newcommand{\up}[1]{\left\lceil #1\right\rceil}
\newcommand{\down}[1]{\left\lfloor #1\right\rfloor}
\newcommand{\R}{\mathbb{R}}
\newcommand{\Z}{\mathbb{Z}}
\newcommand{\AAA}{\mathcal{A}}
\newcommand{\BBB}{\mathcal{B}}
\newcommand{\CCC}{\mathcal{C}}
\newcommand{\nnn}{\hat{n}}

\newcommand{\Thalf}{T_{\mbox{\scriptsize\rm halve}}}
\newcommand{\Tcount}{T_{\mbox{\scriptsize\rm count}}}

\newcommand{\bad}{\text{bad}}

\newcommand{\Pconv}{P_{\mbox{\scriptsize\rm subset-conv}}}
\newcommand{\Qconv}{Q_{\mbox{\scriptsize\rm subset-conv}}}

\newcommand{\Mselect}[2]{M_{\mbox{\scriptsize\rm selected-entries}}(#1\mid #2)}

\begin{document}

\title{3SUM Made Simpler (Without Matrix Multiplication),\\
$k$-SUM Hypothesis Refuted for Odd $k$,\\ and
Real 3SUM/APSP Derandomized}

\author{Timothy M. Chan\thanks{Siebel School of Computing and Data Science, University of Illinois at Urbana--Champaign, USA (\{tmc, itakura3, zm25, yyu51\}@illinois.edu)}\and
Reina Itakura$^\ast$ \and
Zhiyuan Ma$^\ast$ \and
Yuancheng Yu$^\ast$
}

\maketitle

\begin{abstract}
Recently, a truly subquadratic algorithm for the 3SUM problem and
a truly subcubic algorithm for the all-pairs shortest paths (APSP) problem
were found by Anthropic, as reported by Alman and Vassilevska Williams (2026). 
These algorithms relied on a new subroutine for 
rectangular matrix multiplication with selected output entries, which was obtained by modifying known matrix multiplication techniques of Sch\"onhage (1981) and Coppersmith (1982).
\begin{enumerate}
\item We present a different truly subquadratic randomized algorithm
for 3SUM for integer inputs.  It completely avoids the use of matrix multiplication but
instead modifies known dynamic programming algorithms for subset convolution and the zeta transform, which rely only on simple subtraction tricks via inclusion--exclusion (and does not even require fast Fourier transform), and so is conceptually easier to understand (the whole solution is about 5 pages long).
\item Following a similar approach, we obtain a randomized algorithm for $k$-SUM  for integer inputs for odd constant $k$ in $O(n^{\lceil k/2\rceil-\eps_k})$ time for some constant $\eps_k>0$, disproving the $k$-SUM Hypothesis for all odd $k$.
\item Existing reductions from 3SUM and APSP  to selected-entries rectangular matrix multiplication were randomized in the case of real-valued input.
We present new deterministic reductions, yielding
a truly subquadratic deterministic algorithm for Real 3SUM and a truly subcubic deterministic algorithm for Real APSP.
\end{enumerate}
\end{abstract}

\section{Introduction}

Recently, Anthropic found the first truly subcubic algorithm for the all-pairs shortest paths (APSP) problem
and the first truly subquadratic algorithm for the 3SUM problem, as announced and presented 
by Alman and Vassilevska Williams~\cite{AlmanV26}.
The key to the breakthrough lies in a faster algorithm
for the following matrix problem: computing $w$ user-selected (also called ``wanted'') entries of the
product of an $n\times d$ and $d\times n$ matrix, when $w$ is subquadratic.
Let $\Mselect{n,d,n}{w}$ denote its time complexity.
Rephrased in graph-theoretical terms, this rectangular matrix problem includes a lopsided version of the ``\#All-Edges (AE) Sparse Triangle'' problem:
we are given a tripartite graph where the left and right parts $V_L$ and $V_H$ have $n$ vertices and the
middle part $V_M$ has $d$ vertices, and there are $w$ edges between $V_L$ and $V_H$;
we want to count the number of triangles containing $e$, for each edge $e$ between $V_L$ and $V_R$.
Their main subroutine established, for any sufficiently small constants $q,\eps>0$,
the existence of some $\gamma=\gamma(q,\eps)$ such that\footnote{The $\OO$ notation hides polylogarithmic factors.}
\[ \Mselect{n,d,n}{w} \le \OO(n^2/d^\gamma + wd^q) \]
for all $d\le n^\eps$.
For example, $\Mselect{n,d,n}{w} \le 
\OO(n^2/d^{0.063} + wd^{0.437})$ for all $d\le n^{1/18}$. 

The new subroutine was obtained by modifying known matrix multiplication techniques of
Sch\"onhage~\cite{Schonhage81} and Coppersmith~\cite{Coppersmith82} to achieve output-sensitivity.
This is surprising, since matrix multiplication is usually used as a black box in algorithmic applications, and there were not many successful prior examples when modifying matrix multiplication algorithms helped.

Once this matrix problem has been solved, faster algorithms for both APSP and 3SUM---and, more strongly, the
Exact Triangle problem---follow more or less immediately, by using or adapting existing reductions from APSP/3SUM/Exact Triangle to
AE Sparse Triangle \cite{Patrascu10,KopelowitzPP16,WilliamsX20,ChanWX22,ChanX24,FischerKP24}.  There are multiple options, and some of these reductions already hold in the lopsided setting.
In the integer input case, deterministic reductions were known \cite{ChanX24,FischerKP24}; however,
in the real case, only randomized reductions were known \cite{ChanWX22}.
(See also \cite{SahaY26} for a more recent reduction from APSP to selected-entries matrix multiplication in the integer case.)

Several improvements have been found just days after the initial announcement of the breakthrough\footnote{For example,
see  \cite{Maj26} and \url{https://bigospeedrun.com/3sum}.
} 
but in this new era, conceptual understanding and better exposition are perhaps more important than small incremental improvements.
In this manuscript, we address two themes: 1.~simplifying these algorithms, to make the ideas more broadly accessible, and 2.~derandomizing
the algorithms in the real input case.

\paragraph{3SUM simplified.}
Shortly after the announcement of the breakthrough, Karingula and Lovett~\cite{KarL26} have already 
presented a simplification (found with the help of swarms of agents) 
of the main subroutine for selected-entries rectangular matrix multiplication (with slightly worse exponents).
The original explanation of the subroutine~\cite{AlmanV26}, requiring about 18 pages,
 has now been cut down to 7 pages (not counting the remaining reduction from 3SUM/APSP to lopsided AE Sparse Triangle, which would require at least 2 more pages).

While the usage of matrix multiplication tools for APSP is inevitable as APSP is equivalent to
$(\min,+)$-matrix multiplication,
an interesting question is whether further simplifications are possible for 3SUM (to the general public, 3SUM seems to attract even more attention than APSP).
Our intuition is this: 3SUM and the easier $(\min,+)$-convolution problem are about sequences of numbers,
whereas Exact Triangle and $(\min,+)$-matrix multiplication are their matrix counterparts (see Figure~\ref{fig:reduce}).
Sequence problems are more often solved using convolutions or fast Fourier transform (FFT) instead 
of matrix multiplication.  For example, known algorithms for the bounded monotone $(\min,+)$-convolution problem~\cite{ChanL15,ChiDXZ22} 
mostly used FFT, whereas known algorithms for bounded monotone $(\min,+)$-matrix multiplication~\cite{BringmannGSW19,ChiDXZ22} used
fast matrix multiplication.

A natural thought is to try to modify existing algorithms for standard convolution
such as FFT, Karatsuba's, and Toom and Cook's, which if doable might be
easier than modifying matrix multiplication algorithms such as Sch\"onhage's and Coppersmith's,
to those who are not experts in matrix multiplication machineries.  Even if possible, it still seems to require effort to achieve output-sensitivity: for all these recursive algorithms, it is nontrivial to
analyze which nodes of the recursion tree  are
necessary for the selected output entries (see all the figures needed in the descriptions by both
Alman and Vassilevska Williams~\cite{AlmanV26} and Karingula and Lovett~\cite{KarL26}).

We find a simpler way to obtain a truly subquadratic, Monte-Carlo randomized algorithm for 3SUM for polynomially bounded integer inputs. 
The key idea is to switch to a different but related problem, \emph{subset convolution} (see Problem~\ref{prob:subset:conv}).
Standard algorithms for subset convolution are based on the so-called \emph{zeta (or M\"obius) transform}, which
is solved by simple dynamic programming (DP), using just subtraction or inclusion--exclusion tricks,
and are simpler than FFT\@.  
Adapting these subset convolution algorithms to the selected-entries setting turns out to be easy, due to the simple DAG or lattice structure of the DP
(no messy analysis of dependency patterns of ``input encoding/preparation'' and ``output decoding/reconstruction'' in recursion trees~\cite{AlmanV26,KarL26}).
In fact, this observation was noted explicitly before by Bj{\"{o}}rklund, Husfeldt, Kaski, and  Koivisto~\cite{BjorklundHKK10} (for a different application, concerning exponential-time algorithms for certain problems in bounded-degree graphs).\footnote{
Alman and Vassilevska Williams~\cite{AlmanV26} actually cited Bj{\"{o}}rklund, Husfeldt, Kaski, and Mikko Koivisto's work in a remark at the end of their Section 2
on selected-entries matrix multiplication: 
``The pruned recursion of Section 2.4.2 is an analog of FFT pruning~\cite{Markel71,SorensenB93}
of trimmed M\"obius inversion~\cite{BjorklundHKK10}\ldots''\ \ 
However, their paper did not pursue
further the usage of \cite{BjorklundHKK10} to solve 3SUM more directly and more simply.
}
For completeness, we include a description of the selected-entries subset convolution algorithm in Section~\ref{sec:subset:conv}, which takes just one and a half page.  
We believe this is the simplest way to introduce to a broader audience the possibility of output-sensitivity
for selected-entries versions of algebraic problems.\footnote{
We have found that an AI-generated (and Lean-verified) solution posted at \url{https://bigospeedrun.com/3sum} by {\tt c0ldzy17}
also appeared to use subset convolution, but 
there are no accompanying human-written manuscripts explaining the ideas. 
}

While this greatly simplifies the algebraic part of the solution, 
a bit more work is now required in the remaining part, to reduce 3SUM to subset convolution, since subset convolution 
is farther away from 3SUM than standard convolution.  Nevertheless, this part is in more familiar
territory to algorithm designers in fine-grained complexity, and indeed can be solved in under 3 pages using hashing techniques (see Section~\ref{sec:3sum:to:subset:conv}).  
Hashing is commonly used in prior work about convolutions and 3SUM-related reductions (e.g., see 
\cite{ColeH02,BaranDP08,Patrascu10,ChanH20,BringmannFN21,BringmannFN22}), but this reduction requires some nontrivial technical ideas, 
and does not seem to follow directly from any existing work.\footnote{
Jesper Nederlof (personal communication, Nov.~2, 2026) mentioned that he had an unpublished note (circulated to members of a 2022 workshop) describing some similar ideas in reducing min-sum subset convolution to min-sum convolution.
}

\begin{figure}
\centering
\includegraphics[scale=0.8]{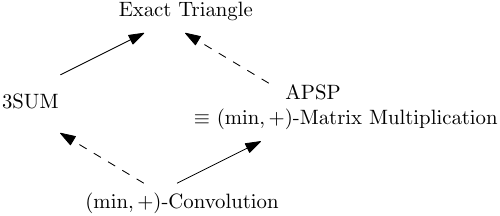}
\caption{Known simple reductions \cite{Patrascu10,BremnerCDEHILPT14,WilliamsW18,CyganMWW19}.  (The dashed ones are for integer input, unless one uses more complicated, recent techniques~\cite{fischer20263sumreallyhardrealtointeger}.)}\label{fig:reduce}
\end{figure}

The solution here raises some new interesting questions concerning whether the 3SUM algorithm here might be considered 
combinatorial, and more generally what is the meaning of ``combinatorial algorithms''.
See Section~\ref{sec:discuss} for discussion.  (Spoiler: we, the authors, do not have a definitive answer, and will
leave the question to the readers to decide for themselves.)

By known simple reductions~\cite{CyganMWW19}, we also immediately obtain
a simpler truly subquadratic randomized algorithm for $(\min,+)$-convolution for polynomially bounded integer inputs, which in turn has many applications, e.g., see \cite{CyganMWW19,KunnemannPS17,BringmannN21,Bringmann24,Jin24}.

\paragraph{$k$-SUM Hypothesis refuted for all odd $k$.}
As an unintended byproduct of this simpler approach, we obtain an $O(n^{\up{k/2}-\eps_k})$-time randomized algorithm for the $k$-SUM problem for polynomially bounded integer inputs, for every odd constant $k\ge 5$, for some constant $\eps_k>0$.
This thus refutes the $k$-SUM Hypothesis for all odd $k$ (at least for the integer setting, with randomization);
see \cite{Erickson99,AbboudL13,AbboudBBK17,KunnemannN22} for conditional lower bounds based on this hypothesis.
Alman and Vassilevska Williams~\cite{AlmanV26} explicitly listed the $k$-SUM Hypothesis for larger $k$ as one of the remaining open questions, as their approach based on selected-entries rectangular matrix multiplication does not seem to generalize.  This demonstrates another advantage of working with convolution-based approaches instead of matrix multiplication.

As is easy to see, $k$-SUM for odd $k$ reduces to an unbalanced instance of 3SUM, where the first  set has $n^{\down{k/2}}$ elements, the second set has $n$ elements, and the third (``query'') set has
$n^{\down{k/2}}$ elements.
As the number of queries is large, we need to refine the 3SUM algorithm to achieve lower cost per query.
To this end, the extra idea is to use convolution where elements are vectors in $[r]^d$ instead of subsets (bit vectors), for a larger constant base $r$.  This extended problem with selected entries can still be solved by adapting standard algorithms for subset convolution.  The remaining hashing part proceeds similar to before, and the analysis is only slightly more involved. (See Section~\ref{sec:ksum}.)

\paragraph{Real 3SUM and APSP derandomized.}
In the second half of the paper (Sections~\ref{sec:real}--\ref{sec:real:3sum}, which can be read independently of Sections~\ref{sec:3sum:simplified}--\ref{sec:ksum}), we turn to the question of derandomization in the real input case 
(the recent algorithms in the integer input case are already deterministic).
The subroutine for selected-entries rectangular matrix multiplication or lopsided AE Sparse Triangle
was already deterministic~\cite{AlmanV26}.
The issue lies in the reduction
of APSP/3SUM to AE Sparse Triangle. 
In the real case, the only known
reduction was the one by Chan, Vassilevska Williams, and Xu~\cite{ChanWX22}, which crucially
relied on random sampling.  
(Fischer, Polak, and Schmidt~\cite{fischer20263sumreallyhardrealtointeger} recently gave a reduction from Real 3SUM to Integer 3SUM, but it was complicated and also used randomization.)
In contrast, the previous best slightly subcubic real APSP algorithm~\cite{Williams18}
before the recent breakthrough was successfully derandomized \cite{ChanW21}, while the previous best slightly subquadratic real 3SUM algorithm~\cite{Chan20} was already deterministic.

We describe deterministic reductions from both Real 3SUM and Real APSP to \#AE Sparse Triangle.  The new reductions build on
Chan, Vassilevska Williams, and Xu's approach~\cite{ChanWX22}; the main new idea is clever but simple, especially for Real APSP (the details are a bit trickier for Real 3SUM\@).
These deterministic reductions lead to the first truly subquadratic and subcubic-time deterministic algorithms for Real 3SUM and Real APSP respectively.



\section{3SUM Simplified}\label{sec:3sum:simplified}

In this section, we provide a truly subquadratic-time randomized algorithm for the 3SUM problem for integer input.
Throughout the section, we use the following notation:
Let $[d]=\{0,\ldots,d-1\}$.
For a set $S$, let $2^S$ denote the set of all subsets of $S$, and
${S\choose k}$ denote the set of all subsets of $S$ of size $k$,
and ${S\choose \le k}$ denote the set of all subsets of $S$ of size at most $k$,
and  ${S\choose \ge k}$ denote the set of all subsets of $S$ of size at least $k$.
Let ${n\choose \le k} = {n\choose 0}+\cdots+{n\choose k}$ and
${n\choose \ge k} = {n\choose k}+\cdots+{n\choose n}$.
Let $1_E$ be the indicator function for the predicate $E$.
We use $\lg$ for base-2 logarithms when the base matters.

\subsection{Selected-Entries Subset Convolution}\label{sec:subset:conv}

In this subsection, we solve a version of the \emph{subset convolution} problem
with selected output entries, formally defined below.  We restrict ourselves to the
version where all input sets have size equal to a given value~$t$ (this is sometimes called 
the ``ranked'' version of the problem).

\begin{problem}[\sf Subset Convolution with Selected Entries]\label{prob:subset:conv}
Let $d\ge 2t$.
We are given collections $\AAA$ and $\BBB$ of subsets in ${[d]\choose t}$,
where each subset $A\in\AAA$ has weight $f(A)\in\R$ and each subset $B\in\BBB$ has
weight $g(B)\in\R$.  By default, let $f(A)=0$ if $A\not\in\AAA$ and $g(B)=0$ if $B\not\in\BBB$.
We are also given a collection $\CCC$ of ``query'' subsets in ${[d]\choose 2t}$.
For each $C\in\CCC$, we want to compute
\[ (f\star g)(C)\ :=\ \sum_{A\subseteq C} f(A)g(C\setminus A) \ =\ \sum_{(A,B):\ A\cap B=\emptyset,\ A\cup B=C} f(A)g(B).
\]
\end{problem}

The standard dynamic programming (DP) algorithm for
subset convolution (sometimes called Yates' algorithm or the ``fast zeta transform'')~\cite{BjorklundHK09,BjorklundHKK07}  turns out to be adaptable to the version with selected entries
and achieves output-sensitive running time, essentially without extra effort!
Roughly, the computation of the DP forms a DAG or lattice.
In the preprocessing phase, we simply simulate DP up to a certain level $s$ of the DAG for a trade-off parameter $s$. When given a query set $C$, we continue the DP from level $s$ up to $2t$, restricting ourselves to nodes of the DAG that are ``under'' $C$.  This ``trimmed'' variant of the fast zeta transform has appeared before in a paper by Bj{\"{o}}rklund, Husfeldt, Kaski, and  Koivisto~\cite{BjorklundHKK10}.  To be self-contained, we give
a complete description below:

\begin{theorem}\label{thm:subset:conv}
Let integer $s$ be a parameter with $0\leq s\leq 2t$.
We can solve Problem~\ref{prob:subset:conv}
in total deterministic time
\[ \left(n+{d\choose \le s} + |\CCC|{2t\choose \ge s}\right) d^{O(1)}.
\]
\end{theorem}
\begin{proof}
We begin with preliminary definitions and facts (see \cite{BjorklundHK09,BjorklundHKK07} and the references therein for more background):
Given $f: 2^{[d]}\rightarrow\R$, the \emph{zeta (or M\"obius) transform} $\widehat{f}: 2^{[d]}\rightarrow\R$ is defined as
\[ \widehat{f}(X) = \sum_{S\subseteq X} f(S).
\]
Given $f,g: 2^{[d]}\rightarrow\R$, the \emph{covering product} $f\star_c g: 2^{[d]}\rightarrow \R$ is defined as
\[ (f\star_c g)(X) = \sum_{(A,B): A\cup B=X} f(A)g(B).
\]
Note that for $f,g:{[d]\choose t}\rightarrow\R$ 
and any $C\in {[d]\choose 2t}$, we have $(f\star g)(C)= (f\star_c g)(C)$,
since $A\cup B=C$, $|A|=|B|=t$, and $|C|=2t$ together imply $A\cap B=\emptyset$.
The covering product satisfies the following known ``inclusion--exclusion'' identity
\cite{BjorklundHKK07}\footnote{
It can be viewed as taking the pointwise product of the zeta transforms of $f$ and $g$, and then taking the
inverse zeta transform (which has a similar definition but with some sign changes).
A quick self-contained proof: the right-hand side expands
to $\sum_{(A,B): A,B\subseteq X} \gamma_{A,B}f(A)g(B)$, where $\gamma_{A,B} := \sum_{S: A\cup B\subseteq S\subseteq X} (-1)^{|X\setminus S|}$ evaluates to 1 if $A\cup B$ is $X$, but is otherwise
equal to $\sum_{i=0}^k (-1)^{k-i} {k\choose i}$, where $k=|X|-|A\cup B|$, which is identical to $(1-1)^k=0$ by the binomial theorem.}:
\[ (f\star_c g)(X) = \sum_{S\subseteq X} (-1)^{|X\setminus S|} \widehat{f}(S)\widehat{g}(S).
\]

\paragraph{Preprocessing algorithm.}  
\begin{enumerate}
\item
For each $X\in {[d]\choose \le s}$, pre-compute
$\widehat{f}(X)$.
More generally, for each $X\in {[d]\choose \le s}$, pre-compute
$f_i(X):=\sum_{S\subseteq X:\ S\cap[i]=X\cap[i]}f(S)$
(noting that $\widehat{f}(X)=f_0(X)$).
To this end, we can use a standard DP formula: 
$f_d(X)=f(X)$, and 
$f_i(X)= f_{i+1}(X) + 1_{[i\in X]} f_{i+1}(X\setminus\{i\})$.
\item
For each $X\in {[d]\choose \le s}$, pre-compute $\widehat{g}(X)$.
More generally, for each $X\in {[d]\choose \le s}$, pre-compute
$f_i(X):=\sum_{S\subseteq X:\ S\cap[i]=X\cap[i]}f(S)$.
This is similar.
\item
For each $X\in {[d]\choose \le s}$, pre-compute
$(f\star_c g)(X)$.
More generally, for each $X\in {[d]\choose \le s}$,
pre-compute $h_i(X):=\sum_{S\subseteq X:\ S\cap[i]=X\cap[i]}(-1)^{|X\setminus S|} \widehat{f}(S)\widehat{g}(S)$
(noting that $h_0(X)=(f\star_c g)(X)$).
To this end,
we use the following DP formula:
$h_d(X)=\widehat{f}(X)\widehat{g}(X)$, and 
$h_i(X)= h_{i+1}(X) - 1_{[i\in X]} h_{i+1}(X\setminus\{i\})$.
\end{enumerate}
The total time\footnote{
We assume the standard word RAM model with word size at least $d$.  If not, the time bounds increase by extra $d^{O(1)}$ factors.
We have ignored cost of dictionary look-up operations for subsets, which is only $d^{O(1)}$.
} for these DPs is $O\left( {d\choose \le s}  d\right)$.

\paragraph{Query algorithm.}
Given a query subset $C\in {[d]\choose 2t}$:
\begin{enumerate}
\item
For each $X\in {C\choose \ge s}$, compute
$\widehat{f}(X)$.
We can use the same DP formula, but using $f_i(X)$ for all $X\in {[d]\choose s}$ (which have been pre-computed during preprocessing) as base
cases.
\item
For each $X\in {C\choose \ge s}$, compute $\widehat{g}(X)$ similarly.
\item
For each $X\in {C\choose \ge s}$, compute
$(f\star_c g)(X)$.
We can use the same DP formula, but using $h_i(X)$ for all $X\in {[d]\choose s}$ (which have been pre-computed during preprocessing) as base
cases.  In particular, we obtain $(f\star_c g)(C)=(f\star g)(C)$ in the end.
\end{enumerate}
The total time per query subset $C$ is
$O\left( {|C|\choose \ge s} d\right) = O\left( {2t\choose \ge s} d\right)$.
\end{proof}

\subsection{Reducing 3SUM to Subset Convolution}\label{sec:3sum:to:subset:conv}

Now we solve integer 3SUM using the new subroutine for selected-entries subset convolution.  This uses
hashing techniques.
We first apply linear hashing (taking ``mod $p$'' for a prime $p$) to reduce the size of the input numbers and then
view the numbers as bit vectors using their binary representations, which then correspond to subsets of indices.
If the subsets corresponding to two numbers happen to be disjoint, then their union corresponds precisely to the sum,
and so subset convolution gives us exactly what we want!
However, usually we do not have disjointness.  The idea is to further randomly scramble the input numbers (using
a family of linear hash functions that achieve 3-wise independence), and with enough repetitions,
the disjointness property can be shown to hold in at least one of the trials with good probability.  

The final algorithm for the reduction is still relatively simple, but the analysis requires some technical effort.  One issue is collision of the hash of sums of pairs (which is why hashing alone does not solve integer 3SUM very efficiently in prior work, except in shaving logs~\cite{BaranDP08}), but in our case, subset convolution filters out pairs that do not satisfy the disjointness property, and this helps to bound the target universe size~$p$.  The precise  details are given below (we state the result more generally when the input sets $A$ and $B$ may not have the same size):

\begin{theorem}\label{thm:3sum:to:subset:conv}
Suppose Problem~\ref{prob:subset:conv} can be solved in
$O(\Pconv(n,d,t) + |\CCC|\Qconv(n,d,t))$ time.

Let $t<\lg \nnn$ be a parameter, and $d=\up{t+\lg\nnn+\lg\lg N+1}$.
Given sets $A,B,C\subseteq [N^{O(1)}]$ with $N=|A|+|B|$ and $\nnn=\sqrt{|A||B|}$,
we can determine whether $c\in A+B$ for each $c\in C$, with error probability $O(1/N)$, in total time
\[ \OO\left(\frac{\nnn^2}{{d\choose 2t}} (\Pconv(N,d,t)+N) + N\frac{\nnn}{2^t}\ +\ |C| \frac{\nnn}{2^t}(\Qconv(N,d,t)+1)\right).
\]
\end{theorem}

\begin{proof}\

\paragraph{Algorithm.} Let $p_0=2^d$.  
Pick a random prime $p$ in $(p_0,2p_0]$.  
Let 
$A_{\bad}=\{a\in A: \exists a'\in A\ \mbox{with}\ a'\equiv a\pmod p\ \mbox{and}\ a'\neq a\}$ and 
$B_{\bad}=\{b\in B: \exists b'\in B\ \mbox{with}\ b'\equiv b\pmod p\ \mbox{and}\ b'\neq b\}$ (these are elements that are in non-singleton buckets when we hash modulo $p$).
For each $c\in C$, if 
$c\in A_\bad+B$ or 
$c\in A+B_\bad$, then report ``yes'' for $c$.

Next, we repeat the following $\mu$ times, for a parameter $\mu$ to be determined:
\begin{enumerate}
\item
Pick 
random numbers $\xi,\alpha,\beta\in [p]$ with $\xi\neq 0$.
Define functions $f,g,h:\Z\rightarrow [p]$:
\begin{eqnarray*}
f(a) &=& (a\xi+\alpha)\bmod p\\
g(b) &=& (b\xi+\beta)\bmod p\\
h(c) &=& (c\xi+\alpha+\beta)\bmod p.
\end{eqnarray*}
Let $\pi:[2^d]\rightarrow 2^{[d]}$ where $\pi(x)$ is the set of all positions of the
1-bits in the binary representation of~$x$.
\item
Let $A^*=\{a\in A\setminus A_\bad: f(a)<2^d,\ |\pi(f(a))|=t\}$ and
$B^*=\{b\in B\setminus B_\bad: g(b)<2^d,\ |\pi(g(b))|=t\}$.
\item
For each $c\in C$ with $h(c)<2^d$ and $|\pi(h(c))|=2t$, compute\footnote{
On $\tau(c)$, this ``second moment'' idea has been used in prior work, e.g., on
sparse convolutions~\cite{BringmannFN21}.
}
\begin{eqnarray*}
\sigma(c) &=& |\{(a,b)\in A^*\times B^*:\ \pi(f(a))\cap\pi(g(b))=\emptyset,\ 
\pi(f(a))\cup\pi(g(b))=\pi(h(c))\}|\\[3pt]
\tau(c) &=& \sum_{(a,b)\in A^*\times B^*:\ \pi(f(a))\cap\pi(g(b))=\emptyset,\ 
\pi(f(a))\cup\pi(g(b))=\pi(h(c))} (a+b-c)^2.
\end{eqnarray*}
If $\sigma(c)>0$ and $\tau(c)=0$, then report ``yes'' for $c$. 
\end{enumerate}
If ``yes'' has not been reported for $c$ in any of the trials, then report ``no''.

\paragraph{Running time analysis.}
Over a random prime $p\in (p_0,2p_0]$, for each $b,b'\in B$,
the probability that $b\equiv b' \pmod p$ is $O((\log N)/p_0)$, since there
are at most $O(\log N/\log p_0)$ prime divisors of $b-b'$ in $(p_0,2p_0]$ and $O(p_0/\log p_0)$ primes in that range.
Thus, 
$\Ex[|A_\bad|]\le \OO(|A|^2/p_0)$, and similarly, 
$\Ex[|B_\bad|]\le \OO(|B|^2/p_0)$.
We can pre-sort $A_\bad+B$ and $A+B_\bad$
in expected time $\OO(|A|^2|B|/p_0+|A||B|^2/p_0)=\OO(N\nnn^2/p_0)=\OO(N\nnn/2^t)$, as $p_0=2^d=\widetilde{\Theta}(\nnn 2^t)$.
Afterwards, the checks for $c\in A_\bad+B$
and $c\in A+B_\bad$ takes $\OO(1)$ time each.


In step~3,
computing the values $\sigma(c)$ corresponds to computing selected entries of the
subset convolution of the collection of sets $\{\pi(f(a)): a\in A^*\}$ 
and $\{\pi(g(b)): b\in B^*\}$, each assigned unit weight.
Computing the values $\tau(c)$ corresponds to computing selected entries 
of $O(1)$ similar subset convolutions, since $(a+b-c)^2$ expands to
$a^2+b^2+2ab-2ca - 2cb + c^2$.
For example, to compute the sum contributed by the $2ab$ term, we assign weight $2a$ to the
set $\pi(f(a))$ and $b$ to the set $\pi(g(b))$.  
The selected entries are $C^* := \{c\in C: h(c)<2^d,\ |\pi(h(c))|=2t\}$.
For fixed $p$, $c$, and $z\in [2^d]$, 
we have $\Pr[h(c)=z]=1/p$ (since for fixed $p,\xi,\alpha$, there is one choice of $\beta$ satisfying the condition).
Thus, $\Pr[c\in C^*] = {d\choose 2t}\frac{1}{p}$.  
So, $\Ex[C^*] = {d\choose 2t}\frac{1}{p}|C|=\Theta\left({d\choose 2t}\frac{1}{\nnn 2^t}|C|\right)$.
The total expected cost
is $\OO\left(\mu\cdot \left(\Pconv(N,d,t)+ N +|C|{d\choose 2t}\frac{1}{\nnn 2^t}\Qconv(N,d,t)\right)\right)$.
Eventually, we will set $\mu=\OO(\nnn^2/{d\choose 2t})$ and this will give the time bound stated in the theorem.
(Expected time bounds can be made worst-case for Monte Carlo algorithms by standard tricks, namely, reporting error when actual time is more than a constant factor times the expectation, and then repeating to lower the error probability.)

\paragraph{Error probability analysis.}
From the definitions of $f,g,h,\pi$, we have:
\begin{itemize}
\item $h(a+b)\equiv f(a)+g(b)\pmod p$.
\item If $\pi(x)\cap\pi(y)=\emptyset$, then $\pi(x+y)=\pi(x)\cup\pi(y)$.
\item $h(c)=h(c')$ iff $c\equiv c'\pmod p$.
\item $\pi(x)=\pi(y)$ iff $x=y$.
\end{itemize}
In particular, if $f(a),g(b)<2^d$ and $\pi(f(a))\cap\pi(g(b))=\emptyset$,
then $\pi(f(a))\cup\pi(g(b))=\pi(h(c))$ iff $\pi(f(a)+g(b))=\pi(h(c))$ iff
$f(a)+g(b)=h(c)$ iff $h(a+b)=h(c)$ iff $a+b\equiv c\pmod p$.

Fix $c\in C$.
Suppose $\sigma(c)>0$ and $\tau(c)=0$.  If a sum of squares with at least one term is equal to 0, then one of the terms (in fact, all of the terms) must be 0.  Thus,
we must have $c=a+b$ for some $a\in A$ and $b\in B$.

Conversely, suppose $c=a+b$ for some $a\in A$ and $b\in B$.
Assume 
$b\not\in B_\bad$ (for otherwise, we would have initially 
reported ``yes'' for $c$).
We will show that $\sigma(c)>0$ and $\tau(c)=0$ with good probability.
To this end, let $L=\{(a',b')\in A\times B: a'+b'\equiv c\pmod p,\ a'+b'\neq c\}$.
Note that the condition $a'+b'\neq c$ implies $a'\neq a$, because $b\not\in B_\bad$.
For a random prime $p\in(p_0,2p_0]$,
for each $(a',b')\in A\times B$ with $a'+b'\neq c$, the probability that $a'+b'\equiv c\pmod p$ is $O((\log N)/p_0)$.
Thus, $\Ex_p[|L|]\le O(\nnn^2 (\log N)/p_0)$.
Call $p$ \emph{good} for $c$ if $|L| \le (\nnn^2\lg^2N)/p_0$.
By Markov's inequality, the probability that $p$ is not good is at most $O(1/\log N)$.
Fix a good prime $p$ for $c$:
\begin{itemize}
\item For any fixed $a,b$ and $x,y\in[2^d]$, we have $\Pr[(f(a)=x) \wedge (g(b)=y)] = 1/p^2$ (since for a fixed $\xi$, there is one choice of $\alpha$ and $\beta$ satisfying the condition).
\item For any fixed $a,b,a'$ with $a\neq a'$ and $x,y,z\in [2^d]$, we have $\Pr[(f(a)=x)\wedge (g(b)=y)\wedge  (f(a')=z)] \le 1/(p^2(p-1))$ (since there is at most one choice of $\alpha,\beta,\xi$ satisfying the condition).
In particular, $\Pr[f(a')=z\mid  (f(a)=x)\wedge (g(b)=y)]\le 1/(p-1)$.
\item Fix $x,y\in [2^d]$ with $|\pi(x)|=|\pi(y)|=t$ and $\pi(x)\cap\pi(y)=\emptyset$.
If $f(a)=x$ and $g(b)=y$, then $\sigma(c)>0$.
Conditioned on $f(a)=x$ and $g(b)=y$,
the event that $\tau(c)>0$
implies the existence of $(a',b')\in L$ with $|\pi(f(a'))|=|\pi(g(b')|=t$ and
$\pi(f(a'))\cup\pi(g(b'))=\pi(h(c))=\pi(x+y)=\pi(x)\cup\pi(y)$.
Letting $f(a')=x'$, there are at most ${2t\choose t}$ candidates for $x'$ satisfying
$|\pi(x')|=t$ and $\pi(x')\subseteq \pi(x)\cup\pi(y)$.
So, 
\begin{equation}\label{eqn1} \Pr[\tau(c)>0 \mid (f(a)=x)\wedge (g(b)=y)]\ \le\ |L| {2t\choose t} \frac{1}{p-1}\ \le\ {2t\choose t} \frac{\nnn^2\lg^2 N}{p_0^2}\ \le\ \frac14
\end{equation}
by goodness of $p$ and by the choice of $p_0=2^d\approx 2^{t+1} \nnn\lg N$ (as we have chosen $d=\up{t+\lg\nnn+\lg\lg N+1}$).
\item Summing over all ${d \choose 2t} {2t\choose t}$
choices of $x,y\in [2^d]$ with $|\pi(x)|=|\pi(y)|=t$ and $\pi(x)\cap\pi(y)=\emptyset$,
we then obtain
\begin{equation}\label{eqn2} \Pr[(\sigma(c)>0)\wedge (\tau(c)=0)]\ \ge\ {d \choose 2t} {2t\choose t}\frac{1}{p^2}\cdot\frac{3}{4}\ \ge\ \widetilde{\Omega}\left({d\choose 2t}\frac{1}{\nnn^2}\right).
\end{equation}
\end{itemize}
By choosing the number of trials to be $\mu=\OO(\nnn^2/{d\choose 2t})$, we can guarantee that the error probability 
is at most, say, $0.1$, conditioned on a good prime $p$ for $c$.
Unconditionally, the error probability is at most $0.1+O(1/\log N)$ per $c$.
By repeating the entire algorithm $O(\log(N|C|))$ times, the error probability can be
made below $1/(N|C|)$ per~$c$.
\end{proof}

As a sanity check, if we plug in the standard non-output-sensitive bounds
for subset convolution with $\Pconv=\OO({d\choose 2t})$ and $\Qconv=O(1)$ to the above theorem,
we recover the trivial quadratic bound for 3SUM\@.
By applying Theorem~\ref{thm:subset:conv} instead, we obtain an overall time bound of
\[ 
\OO\left(  \frac{\nnn^2}{{d\choose 2t}} \left({d\choose \le s} + N\right) + N\frac{\nnn}{2^t}\ +\ |C|\frac{\nnn}{2^t} {2t\choose \ge s}
\right),
\]
where $d=t+\lg\nnn+O(\log\log N)$.
Setting $t=\gamma\lg \nnn$ and $d=(1+\gamma)\lg\nnn+O(\lg\lg\nnn)$ and $s=(2\gamma-\delta)\lg \nnn$ for $\delta<\gamma<1/3$
yields
\[ \OO\left(  \frac{\nnn^2}{{(1+\gamma)\lg \nnn\choose 2\gamma\lg\nnn}} \left( {(1+\gamma)\lg\nnn\choose (2\gamma-\delta)\lg\nnn} + N\right) + N\nnn^{1-\gamma}\ +\ |C|\nnn^{1-\gamma} {2\gamma\lg\nnn\choose \delta\lg\nnn}
\right).
\]
For a sufficiently small constant $\delta$, this gives a truly subquadratic algorithm.
For example, for $\gamma=0.2$ and $\delta=0.033$, we obtain:

\begin{corollary}
Given sets $A,B,C\subseteq [N^{O(1)}]$ with $N=|A|+|B|$ and $\nnn=\sqrt{|A||B|}$,
we can determine whether $c\in A+B$ for each $c\in C$, with error probability $O(1/N)$,
in $O(\nnn^{1.965} + N\nnn^{0.899}\ +\ |C|\nnn^{0.965})$ total time.
\end{corollary}

For $|A|=|B|=|C|=n$, this gives an $O(n^{1.965})$-time algorithm for 3SUM (or more generally ``All-Numbers 3SUM''~\cite{WilliamsW18})
for integer input.  Though it is not our intention,
this improves over the $O(n^{1.998})$ randomized time bound originally reported by Alman and Vassilevska Williams~\cite{AlmanV26}.\footnote{
The current best (Lean-verified) result reported at \url{https://bigospeedrun.com/3sum} as of Oct 21, 2026 is near $O(n^{1.921})$, but there is no accompanying human-written manuscript.  We reiterate that our main goal is better exposition, to make the ideas easier to understand and appreciate.
}

In fact, the approach gives a data structure for preprocessing $A$ and $B$ in $O(n^{1.965})$ time and space, so that we can determine whether a query number is in $A+B$ in $O(n^{0.965})$ 
time with low error probability (this data structure problem is called ``3SUM Indexing''~\cite{GoldsteinKLP17}).
It is easy to separate the preprocessing and the query parts of the algorithm in Theorem~\ref{thm:3sum:to:subset:conv}.

It is likely possible to turn this Monte Carlo randomized algorithm into Las Vegas with more effort, but we will not elaborate as we aim for the simplest presentation.




\subsection{Discussion}\label{sec:discuss}

The algorithm here raises some interesting  philosophical questions:

\paragraph{Could the breakthrough have been found earlier?}
To some researchers, the solution reported by Alman and Vassilevska Williams~\cite{AlmanV26} felt like something that
humans could have eventually discovered, or could have been discovered earlier had we been luckier.  It might be that not many researchers in fine-grained complexity
are comfortable enough to open up the matrix multiplication black box to seriously try
solving the selected-entries matrix multiplication problem faster.
The approach here completely lowers the barrier of entry on the algebraic side (subset convolution 
is only about inclusion--exclusion formulas), so
in hindsight, there is even less excuse for researchers to have missed the breakthrough!
Experts in subset convolution~\cite{BjorklundHKK10} already knew the idea in Section~\ref{sec:subset:conv}, while experts in hashing might be able to figure out the details in Section~\ref{sec:3sum:to:subset:conv} on their own, with some effort.
However, without the more complicated solution in~\cite{AlmanV26}, we would not have seriously attacked the 3SUM problem
via selected-entries subset convolution in the first place.

\paragraph{Is the 3SUM algorithm here combinatorial?}
The answer to this question is subjective, as there is no formal definition of ``combinatorial algorithms''.
\begin{itemize}
\item If ``combinatorial'' means avoiding fast matrix multiplication but allowing FFT, the answer is yes.
The motivation for this informal definition is
 that fast matrix multiplication algorithms are not very practical, but FFT is.
\item If ``combinatorial'' means avoiding both fast matrix multiplication and FFT, the answer is still yes.
We use the fast zeta transform, which is just dynamic programming with inclusion--exclusion formulas.
Zeta transforms and inclusion--exclusion have been used before not just in exponential-time algorithms (e.g., \cite{Karp82,BjorklundHKK07,BjorklundHK09}) but also
for some counting problems in fine-grained complexity (e.g.,~\cite{WilliamsW13}).
In algorithm design, we often use subtraction tricks without referring to them as non-combinatorial.
(In fact, the word ``combinatorial'' literally refers to counting, and inclusion--exclusion is a standard trick in combinatorics.)
However, one could debate why we allow inclusion--exclusion but not, say, Strassen's matrix multiplication
algorithm (or even Sch\"onhage's), which is also based on ``just'' subtractions and clever cancellations. 
\item If ``combinatorial'' means avoiding completely all algebraic techniques, the answer is no.
However, one could say that hashing techniques, commonly used in algorithm design (including the reduction
in Section~\ref{sec:3sum:to:subset:conv}), are also algebraic. Then many algorithms that were considered combinatorial before
might not be.  It is not clear where one would draw the line\ldots
\end{itemize}

If the algorithm here were considered to be combinatorial, then
this would also imply a new combinatorial, truly subquadratic algorithm for Boolean convolution
of two $n$-bit vectors, since the latter is a special case of All-Numbers 3SUM\@;
but this algorithm would be more complicated (and slower) than, say, Karatsuba's algorithm.

If the ultimate goal is just to exclude non-practical algorithms, then the current consensus of combinatorial algorithms does not seem too appropriate: some of the known ``combinatorial'' algorithms such as
\cite{AbboudFKLM24} for Boolean matrix multiplication are far more complicated and probably less practical than Strassen's, for example, and some of the known algorithms for shaving logs \cite{Chan15,Yu18} have large hidden constant factors and would not be competitive in practice.  The 3SUM algorithm here in its present form may not be too practical either.

If the defining characteristic of combinatorial matrix multiplication algorithms is generalizability 
to $(\min,+)$-matrix multiplication (which would successfully exclude Strassen's but include some log-shaving
algorithms~\cite{Chan19}), then the recent breakthrough on
APSP/$(\min,+)$-matrix multiplication has turned the question upside-down.

Here is another alternative, informal, and imprecise way to characterize combinatorial algorithms: \emph{intermediate values
maintained by the algorithm should be describable ``naturally''}.  For example, Strassen's matrix multiplication
algorithm and Karatsuba's convolution algorithm do not seem to fulfill this requirement, since
intermediate values arising during recursion seem ``magical'' and do not have clear ``meanings''.  On the other hand, 
the 3SUM algorithm here arguably satisfies this criterion: the intermediate values maintained by the DP are the
zeta transforms and the covering products of subsets, which have relatively simple definitions.

\section{More on 3SUM, and $k$-SUM for Odd $k$}\label{sec:ksum}

The 3SUM algorithm in the previous section has query cost $O(n^{0.965})$, which 
cannot be lowered to below $n^{2/3}$ even if we allow larger preprocessing.
In this section, we describe a variant of the algorithm, which can make the query cost
an arbitrarily small polynomial, while still supporting the unbalanced case when the sets $A$ and $B$ have different sizes.  As a consequence, this leads to a faster
algorithm for $k$-SUM for any constant odd $k$ for integer input.

\subsection{Selected-Entries Convolution in $[r]^d$ for Small $r$}

The main additional idea is to use a larger base $r>2$.  This requires
a generalization of the subset convolution problem where elements are  vectors in $[r]^d$.
The $r=2$ or bit-vector case corresponds precisely to subset convolution.  
This is essentially  standard $d$-dimensional convolution with coordinate values in $[r]$, which is known to be reducible to standard one-dimensional convolution over $[O(r)^d]$ when $d$ is a constant and $r$ is a variable.  However, we need the opposite regime when $r$ is a constant and $d$ is a variable, and we need the version with selected entries.

\newcommand{\HHH}{{\cal H}}

\begin{problem}\label{prob:vec:conv}
Let $\HHH_{r,d,t}=\{x\in [r]^d: \|x\|_1=t\}$.
We are given two functions $f:\HHH_{r,d,t_1}\rightarrow\R$ and
$g:\HHH_{r,d,t_2}\rightarrow\R$ with $O(n)$ nonzero entries
and a subset $C\subseteq \HHH_{r,d,t}$ with $t=t_1+t_2$.  For each $c\in C$,
we want to compute
\[ (f\star_r g)(c)\  :=\ \sum_{a,b\in [r]^d:\ c=a+b} f(a)g(b),\]
where $a+b$ denotes vector addition over $\Z^d$.
\end{problem}

This version of the problem for vectors in $[r]^d$ can be viewed as selected-entries subset convolution for 
\emph{multisets} with multiplicities bounded by $r$.  We can reduce it to standard
subset convolution for sets, by reinterpreting multisets in $[d]$ as sets in $[dr]$, as follows.
Given a set $A\subseteq [dr]$, let $\mu_{j,A}=|A\cap [jr,(j+1)r)|$; we call this 
the \emph{multiplicity} of $j$.
Two sets $A,A'\subseteq [dr]$ are \emph{multiplicity-equivalent}
if $\mu_{j,A}=\mu_{j,A'}$ for all $j\in [d]$.
A collection $\AAA$ of sets in $[dr]$ with weight function $f$ is
\emph{multiplicity-invariant} if
for every $A\in\AAA$ and every $A'$ that is multiplicity-equivalent to $A$,
we have $A'\in\AAA$ and $f(A')=f(A)$.
A multiplicity-invariant collection can be represented implicitly: it suffices
to store one representative set in each equivalence class, and for each such set $A$,
store just the multiplicities of each of its elements.
The subset convolution of two multiplicity-invariant collections is multiplicity-invariant,
and the zeta transform of a multiplicity-invariant collection is multiplicity-invariant.

\begin{lemma}\label{thm:subset:conv:mult}
Let $s\le t$.
Given two multiplicity-invariant, weighted collections $\AAA$ and $\BBB$ of sets in $[dr]$ with $O(n)$ multiplicity-equivalence classes, where sets in $\AAA$ have size $t_1$ and
sets in $\BBB$ have size $t_2$, and given
a collection $\CCC$ of sets in $[dr]$ of size $t=t_1+t_2$,
we can compute the subset convolution output of each set in $\CCC$ in total deterministic time
\[ \left(n + H_r(d, \le\! s)\ +\ |\CCC| H_r(d,t,\ge\! s)\right) (dr)^{O(1)},
\]
where $H_r(d,\le\!s)$ denotes the number of vectors in $[r]^d$ with $L_1$ norm at most $s$,
and $H_r(d,t,\ge\!s)$ denotes the maximum number of vectors in $[r]^d$ with $L_1$ norm at least $s$ that is dominated\footnote{
$(x_0,\ldots,x_{d-1})$ is \emph{dominated} by $(y_0,\ldots,y_{d-1})$ if $x_j\le y_j$ for all $j\in [d]$. 
} by a vector of  $L_1$ norm  $t$.
\end{lemma}
\begin{proof}
We use the same algorithm as in the proof of Theorem~\ref{thm:subset:conv}
, with straightforward changes:

In step~1 of the preprocessing algorithm, 
it suffices to compute and store $\widehat{f}(X)$ for just
one representative set $X$ per equivalence class (since $\widehat{f}$ is multiplicity-invariant),
and store $f_i(X)$ for $O(r)$ representative sets $X$ (since the value depends on the multiplicity vector, along with
the cardinality of $X\cap [jr,(j+1)r)$ for the unique $j$ with $i\in [jr,(j+1)r)$).
Steps~2 and 3 are similar.
The number of table entries is $O(drH_r(d,\le s))$, since $H_r(d,\le s)$
bounds the number of equivalence classes of sets of size at most $s$.

The query algorithm is similar.  Here, the number of table entries is $O(drH_r(d,t,\ge s))$,
since $H_r(d,t,\ge s)$ bounds the number of equivalence classes of sets of size at least $s$ contained in the query set $C$.
\end{proof}

\begin{theorem}\label{thm:vec:conv}
Let $s\le t$.  Problem~\ref{prob:vec:conv} can be solved in total deterministic time $\left( n + H_r(d, \le\! s)\ +\ |C| H_r(d,t,\ge\!s)\right) (dr)^{O(1)}$.
\end{theorem}
\begin{proof}
For each vector $a=(a_0,\ldots,a_{d-1})\in [r]^d$ with a nonzero $f(a)$,
additionally define $f'(A)=f(a) \cdot (a_0!\cdots a_{d-1}!)$ for all sets $A\subseteq [dr]$ with $\mu_{j,A}=a_j$ for each $j\in [d]$.
For each vector $b=(b_0,\ldots,b_{d-1})\in [r]^d$ with a nonzero $g(b)$,
define $g'(B)=g(b) \cdot (b_0!\cdots b_{d-1}!)$ for all sets $B\subseteq [dr]$ with $\mu_{j,B}=b_j$ for each $j\in [d]$.
Then $f'$ and $g'$ are multiplicity-invariant.
For $c=(c_0,\ldots,c_{d-1})\in [r]^d$, we have 
$(f\star_r g)(c)= (f'\star g')(c) / (c_0!\cdots c_{d-1}!)$.
\end{proof}

\subsection{Reducing 3SUM to Convolution in $[r]^d$}

\newcommand{\Pvecconv}{P_{\mbox{\scriptsize\rm vec-conv}}}
\newcommand{\Qvecconv}{Q_{\mbox{\scriptsize\rm vec-conv}}}

Now we modify the algorithm in Section~\ref{sec:3sum:to:subset:conv} 
using convolution over $[r]^d$.  When working in base $r$, 
we no longer seek disjointness of subsets corresponding to two numbers.  Instead, we just want the two numbers to be
summable without carries, so that summing the two numbers
is equivalent to summing
their corresponding vectors.
With a larger base $r$, this turns out to happen with greater probability.

\begin{theorem}\label{thm:3sum:to:vec:conv}
Suppose Problem~\ref{prob:vec:conv} can be solved in
$O(\Pvecconv(n,d,t,r) + |\CCC|\Qvecconv(n,d,t,r))$ time.

Let $r$ be a constant.
Given sets $A,B,C\subseteq [n^{O(1)}]$ with $N=|A|+|B|$ and $\nnn=\sqrt{|A||B|}$,
we can determine whether $c\in A+B$ for each $c\in C$, with error probability $O(1/N)$,
in 
\[ \OO\left(\nnn^{2(1-\log_\rho r)}\cdot \big(\Pvecconv(N,d,t,r)+N\ +\ |C| (\Qvecconv(N,d,t,r)+1)\big)\right)\] 
where $t=(r-1)d/2$, and $d=2\log_\rho\nnn+O(\log\log N)$ for $\rho := r^2/(r!)^{1/r}$.  Note that $\rho=\Theta(r)$, and so 
$1-\log_\rho r=O(1/\log r)$.
\end{theorem}

\begin{proof} 
Assume $d$ is divisible by $r$.  We modify the proof of Theorem~\ref{thm:3sum:to:subset:conv}.

\paragraph{Algorithm.} Let $p_0=r^d$.  
Pick a random prime $p$ in $(p_0,2p_0]$.  
Let 
$A_{\bad}=\{a\in A: \exists a'\in A\ \mbox{with}\ a'\equiv a\pmod p\ \mbox{and}\ a'\neq a\}$
and 
$B_{\bad}=\{b\in B: \exists b'\in B\ \mbox{with}\ b'\equiv b\pmod p\ \mbox{and}\ b'\neq b\}$ (these are elements that are in non-singleton buckets when we hash modulo $p$).
For each $c\in C$, if 
$c\in A_\bad+B$ or 
$c\in A+B_\bad$, then report ``yes'' for $c$.

Next, we repeat the following $\mu$ times, for a parameter $\mu$ to be determined:
\begin{enumerate}
\item
Pick 
random numbers $\xi,\alpha,\beta\in [p]$ with $\xi\neq 0$.
Define functions $f,g,h:\Z\rightarrow [p]$:
\begin{eqnarray*}
f(a) &=& (a\xi+\alpha)\bmod p\\
g(b) &=& (b\xi+\beta)\bmod p\\
h(c) &=& (c\xi+\alpha+\beta)\bmod p.
\end{eqnarray*}
Let $\pi:[r^d]\rightarrow [r]^{d}$ where $\pi(x)$ is the vector formed by the base-$r$
representation of~$x$.
\item
Let $A^*=\{a\in A\setminus A_\bad: f(a)<r^d\}$ and
$B^*=\{b\in B\setminus B_\bad: g(b)<r^d\}$.
\item
For each $c\in C$ such that $h(c)<r^d$ and $\pi(h(c))$ is \emph{perfectly balanced}, i.e.,
each number in $[r]$ occurs exactly $d/r$ times in the vector, compute
\begin{eqnarray*}
\sigma(c) &=& |\{(a,b)\in A^*\times B^*:\ \pi(f(a))+\pi(g(b))=\pi(h(c))\}|\\[3pt]
\tau(c) &=& \sum_{(a,b)\in A^*\times B^*:\ \pi(f(a))+\pi(g(b))=\pi(h(c))} (a+b-c)^2.
\end{eqnarray*}
If $\sigma(c)>0$ and $\tau(c)=0$, then report ``yes'' for $c$. 
\end{enumerate}
If ``yes'' has not been reported for $c$ in any of the trials, then report ``no''.

\paragraph{Running time analysis.}
The analysis of the algorithm is similar, so we will only point out the differences.

Initially, sorting $A_\bad+B$ and $A+B_\bad$  takes
$\OO(N\nnn^2/p_0) = \OO(N\nnn^{2(1-\log_\rho r)})$ time, as $p_0=r^d$ and 
$d=2\log_\rho\nnn + O(\log\log N)$.

In step~3, for the convolutions, we have to loop through all combinations of $t_1,t_2$
with $t_1+t_2=t$, but this only increases the running time by an $O(t)$ factor.

Note that a perfectly balanced vector has $L_1$ norm $(d/r)\sum_{i=0}^{r-1}i=(r-1)d/2$, and
the number of perfectly balanced vectors is equal to the multinomial coefficient $\frac{d!}{((d/r)!)^r} = r^d/d^{\Theta(r)}=\widetilde{\Theta}(r^d)$.
The selected entries are $C^*$ are the vectors in $C$ that are perfectly balanced, and so $\Ex[C^*] = \widetilde{\Theta}(\frac{r^d}{p}|C|)=\widetilde{\Theta}(|C|)$.
So, the running time now becomes
$\OO\left(\mu\cdot \left(\Pvecconv(N,d,t,r)+N +|C|
\Qvecconv(N,d,t,r)\right)\right)$.
Eventually we will choose $\mu = \OO(\nnn^{2(1-\log_\rho r)})$, and this will give the time bound stated in the theorem.

\paragraph{Error probability analysis.}
The mapping $\pi$ satisfies the property: if
$\pi(x)+\pi(y)\in [r]^d$, then $\pi(x+y)=\pi(x)+\pi(y)$ (this corresponds to ``no-carry'' addition of base-$r$ numbers).

Observe that if $z\in [r]^d$ is a perfectly balanced vector,
the number of vectors $x'\in [r]^d$ dominated by $z$ is
equal to $\prod_{i=0}^{r-1} (i+1)^{d/r} = (r!)^{d/r}$.
Equation (\ref{eqn1}) then becomes 
\[ \Pr[\tau(c)>0 \mid (f(a)=x)\wedge (g(b)=y)]\ \le\
|L|\cdot (r!)^{d/r} \frac{1}{p_0}\ \le\ (r!)^{d/r}\cdot \frac{\nnn^2\lg^2 N}{p_0^2}\ \le\ \frac{1}{4}
\]
by choosing $d$ so that $p_0^2=r^{2d} \approx  4(\nnn^2\lg^2 N) (r!)^{d/r}$.
This implies that $(r^2/(r!)^{1/r})^d = \widetilde{\Theta}(\nnn^2)$.
So, $d=2\log_\rho\nnn+O(\log\log N)$ where $\rho=r^2/(r!)^{1/r}$.

Equation (\ref{eqn2}) becomes
\[ \Pr[(\sigma(c)>0)\wedge (\tau(c)=0)]\ \ge\ \frac{d!}{((d/r)!)^r} (r!)^{d/r} \frac{1}{p^2}\cdot\frac{3}{4}\ \ge\
\widetilde{\Omega}\left(\frac{d!}{((d/r)!)^r}\cdot\frac{1}{\nnn^2}\right)\ \ge\
\widetilde{\Omega}\left( \frac{r^d}{\nnn^2}\right).
\]
We now choose $\mu = \OO\left(\frac{\nnn^2}{r^d}\right)=\OO(\nnn^{2(1-\log_\rho r)})$.  Remaining parts of the proof are similar.
\end{proof}

Applying Theorem~\ref{thm:vec:conv}, we obtain an overall time bound of 
\[ \OO\left(\nnn^{2(1-\log_\rho r)}\cdot \big(H_r(d,\le\!s)+N\ +\ |C| (H_r(d,\ge\!s,(r-1)d/2)+1)\big)\right)\] 
for $d=2\log_\rho\nnn+O(\log\log N)$.

Let $s=(1-\delta)(r-1)d/2$ for a sufficiently small constant $\delta>0$.
Observe that $H_r(d,\le\!s)\le r^d 2^{-\Theta(\delta^2 d)}= \widetilde{\Theta}(\nnn^{2\log_\rho r- \Theta(\delta^2/\log\rho)})$
(since the probability that a sum of $d$ independent, uniform random variables in $[r]$
is at most $(1-\delta)(r-1)d/2$ is at most $2^{-\Theta(\delta^2 d)}$ by a Chernoff bound).
On the other hand,
$H_r(d,\ge\!s,(r-1)d/2)$ is bounded by ${(r-1)d/2\choose \ge s} \le {(r-1)d/2\choose O(\delta rd)}
= (1/\delta)^{O(\delta rd)} = \nnn^{O(\delta r\log (1/\delta)/\log \rho)}$.
Thus, the time bound becomes
\[
O\left( \nnn^{2-\Theta(\delta^2/\log r)} + (N+|C|) \nnn^{O(1/\log r + \delta r\log(1/\delta)/\log r)} \right).
\]
For example, setting $\delta=\eps^{0.49}$ and $r=1/\eps^{0.48}$ implies:

\begin{corollary}
Given sets $A,B,C\subseteq [N^{O(1)}]$ with $N=|A|+|B|$ and $\nnn=\sqrt{|A||B|}$,
we can determine whether $c\in A+B$ for each $c\in C$, with error probability $O(1/N)$,
in $O(\nnn^{2-\eps} + (N+|C|)\nnn^{O(1/\log(1/\eps))})$ time for a sufficiently small constant $\eps>0$.
\end{corollary}

The $(2k+1)$-SUM problem reduces to 3SUM in the unbalanced case when $|A|=n^k$, $|B|=n$, and $|C|=n^k$ (with $N=O(n^k)$ and $\nnn=O(n^{(k+1)/2})$), and so we immediately obtain:

\begin{corollary}
For any constant $k\ge 1$, the $(2k+1)$-SUM problem for polynomially bounded integers
can be solved in $O(n^{k+1-\eps_k})$ for some constant $\eps_k=2^{-\Theta(k)}$ by a Monte Carlo randomized algorithm.
\end{corollary}

\section{Real APSP Derandomized}\label{sec:real}

\newcommand{\RANK}{\mbox{\rm rank}}
\newcommand{\SIZE}{\mbox{\rm size}}
\newcommand{\SUMRANK}{\mbox{\rm sum-rank}}
\newcommand{\AVGRANK}{\mbox{\rm avg-rank}}

In this section, we present a truly subcubic-time deterministic algorithm for APSP for real input.

It is well known (by repeated squaring) that the real APSP problem reduces to computing the
$(\min,+)$-product of two $n\times n$ real matrices, which in turn
reduces to $n/d$ instances of $(\min,+)$-product of $n\times d$ and $d\times n$ matrices:
 
\begin{problem}[\sf Rectangular $(\min,+)$ matrix multiplication]\label{prob:orig}
Given an $n\times d$ real matrix $A$ and a $d\times n$ matrix $B$,
compute $C[i,j]=\min_{k\in[d]}(A[i,k]+B[k,j])$ for each $i,j\in[n]$.
\end{problem}

For simplicity of presentation, we assume that there are no degeneracies (no two pairs of input numbers have equal sum); degenerate cases can be addressed by perturbation or by straightforward modifications.

Chan, Vassilevska Williams, and Xu's previous approach~\cite{ChanWX22} considered an intermediate ``verification'' problem:

\begin{problem}[\sf Verification]\label{prob:decis}
We are given an $n\times d$ real matrix $A$ and a $d\times n$ matrix $B$ 
and indices $p_{ij}\in [d]$ for each $i,j\in[n]$.   For each $i,j\in[n]$, we want to 
decide whether there exists an index
$p'_{ij}\in[d]$ satisfying 
\[ A[i,p'_{ij}]+A[p'_{ij},j] < A[i,p_{ij}]+A[p_{ij},j], \]
and if yes, report one such $p'_{ij}$.
\end{problem}

If the answer is no, then $p_{ij}$ is indeed the correct index for the $(\min,+)$-product entry $C[i,j]$.

Chan, Vassilevska Williams, and Xu showed that this verification problem can be reduced to Lopsided AE Sparse Triangle (and thus to selected-entries rectangular matrix multiplication), by adapting simple ideas found
in previous slightly subcubic APSP algorithms such as \cite{Fredman76,Chan08,Chan10,Williams18} (notably, ``Fredman's trick'').

To reduce the original problem to this verification problem, they applied a well-known random sampling approach: take a random subset $R\subseteq [d]$ of a certain  size $r$, compute $\arg\min_{k\in R}(A[i,k]+B[k,j])$ for each $i,j$ (by brute force or by recursion).  This gives us good candidates $p_{ij}$ for verification, since for each $i,j$, the number of $p'$ with a smaller $A[i,p']+B[p',j]$ is expected to be around $d/r$.

Unfortunately, we do not see a good way to derandomize this sampling approach. (Essentially, we want a small hitting set $R$ for a certain set system which we cannot afford to generate explicitly.)  We have unsuccessfully tried various  sophisticated ideas, such as the method of conditional probabilities, or limited independence.  In the integer case, Chan and Xu~\cite{ChanX24} suggested halving the input values, and recursively solving the problem in a smaller integer universe; this then gives us good candidates for verification.  However, this approach inherently does not work for real input.

\subsection{New Counting Problem}

We describe a simple approach that gives a deterministic reduction. 
We first define the following ``halving'' problem, which is stronger than the verification problem:
 
\begin{problem}[\sf Halving]\label{prob:half}
We are given an $n\times d$ real matrix $A$ and a $d\times n$ matrix $B$ 
and indices $p_{ij}\in [d]$ for each $i,j\in[n]$ with $d\le n$.  For each $i,j\in[n]$, we want to
compute a new index $p'_{ij}\in[d]$ satisfying
\[ \RANK_{ij}(p'_{ij}) \le \up{\RANK_{ij}(p_{ij})/2},
\]
where $\RANK_{ij}(p)$ denotes the rank of $A[i,p]+B[p,j]$ in $\{A[i,k]+B[k,j]: k\in [d]\}$.
\end{problem}
 
By iterating $\log d$ times, we can immediately reduce the rectangular $(\min,+)$ matrix multiplication
problem to the halving problem:
 
\begin{lemma}\label{lem:orig:to:half}
If Problem~\ref{prob:half} can be solved in $\Thalf(n,d)$ time, then
Problem~\ref{prob:orig} can be solved in $O(\Thalf(n,d)\log d)$ time.
\end{lemma}

To find an element with low rank, the key new idea is to find an interval with low
\emph{average} rank.  This motivates the following ``counting'' problem:
 
\begin{problem}[\sf Counting]\label{prob:count}
We are given an $n\times d$ real matrix $A$ and a $d\times n$ matrix $B$,
indices $p_{ij}\in [d]$ for each $i,j\in[n]$,
and dyadic\footnote{
A \emph{dyadic interval} is an interval of the form $(a2^i,(a+1)2^i]$ for integers $a$ and $i$.
Without loss of generality, we may assume that $d$ is a power of 2.
} intervals $I_{ij}\subseteq [d]$ of length at least $d^{1-\delta}$ for a given parameter $\delta>0$.
Let \[ L_{ij}(p,I)\ :=\ \{k\in I: A[i,k]+B[k,j] < A[i,p]+B[p,j]\}. \]
For each $i,j\in [n]$, we want to
compute the following counts:
\begin{eqnarray*}
\SIZE_{ij}(p_{ij},I_{ij}) &:=& |L_{ij}(p_{ij},I_{ij})|\ 
\\[2pt]
\SUMRANK_{ij}(p_{ij},I_{ij}) &:=& \sum_{k\in L_{ij}(p_{ij},I_{ij})} \RANK_{ij}(k)\\
&=& |\{(k,k')\in I_{ij}\times[d]:
A[i,k']+B[k',j]\le A[i,k]+B[k,j] < A[i,p_{ij}]+B[p_{ij},j]\}|.
\end{eqnarray*}
\end{problem}

Note subtly that in the definition of $\SUMRANK_{ij}(p_{ij},I_{ij})$, we are summing
the ``global'' ranks $\RANK_{ij}(k)$ relative to $[d]$, not ranks relative to $I_{ij}$ (and thus
the index $k'$ runs over $[d]$, not $I_{ij}$).  From the above counts, we can then get:
\[ \AVGRANK_{ij}(p_{ij},I_{ij})\ :=\ \SUMRANK_{ij}(p_{ij},I_{ij})/\SIZE_{ij}(p_{ij},I_{ij}).
\]

We now show how to reduce the halving problem to the counting problem by a simple binary search: we just
repeatedly shrink intervals while maintaining low average rank.
 
\begin{lemma}\label{lem:half:to:count}
For any constant $\delta>0$, if Problem~\ref{prob:count} can be solved in $\Tcount(n,d,\delta)$ time, then
Problem~\ref{prob:half} can be solved in $\Thalf(n,d)=O(\Tcount(n,d,\delta)\log d  + d^{1-\delta}n^2)$ time.
\end{lemma}
\begin{proof}
Let $r_{ij}=\RANK_{ij}(p_{ij})$.
Observe that $\SIZE_{ij}(p_{ij},[d])=r_{ij}$, and 
$\SUMRANK_{ij}(p_{ij},[d]) = r_{ij}(r_{ij}-1)/2$, and so
$\AVGRANK_{ij}(p_{ij},[d])\le r_{ij}/2$, assuming $r_{ij}\ge 1$.

For each $i,j\in[n]$, we maintain a dyadic interval $I_{ij}$ such that
$\SIZE_{ij}(p_{ij},I_{ij})\ge 1$ and 
$\AVGRANK_{ij}(p_{ij},I_{ij})\le r_{ij}/2$.  At the beginning, we
can set $I_{ij}=[d]$ for all $i,j$.

In each round, we divide the dyadic interval $I_{ij}$ into its left half $I^-_{ij}$ and  its right half $I^+_{ij}$.
Since $\SIZE_{ij}(p_{ij},I_{ij})=\SIZE_{ij}(p_{ij},I_{ij}^-) + \SIZE_{ij}(p_{ij},I_{ij}^+)$
and $\SUMRANK_{ij}(p_{ij},I_{ij})=\SUMRANK_{ij}(p_{ij},I_{ij}^-) + \SUMRANK_{ij}(p_{ij},I_{ij}^+)$), we must have
$\SIZE_{ij}(p_{ij},I^-_{ij})>0$ and $\AVGRANK_{ij}(p_{ij},I^-_{ij}) \leq r_{ij}/2$, or 
$\SIZE_{ij}(p_{ij},I^+_{ij})>0$ and $\AVGRANK_{ij}(p_{ij},I^+_{ij})\leq r_{ij}/2$.
We can tell which is the case by invoking an algorithm for Problem~\ref{prob:count}.
We then reset $I_{ij}$ to $I^-_{ij}$ in the former case, or $I^+_{ij}$ in the latter.

The total number of rounds is $O(\log d)$.
When the length of $I_{ij}$ drops below $d^{1-\delta}$,
we can return the index $p'_{ij}\in I_{ij}$ minimizing $A[i,p'_{ij}]+B[p'_{ij},j]$,
which satisfies $\RANK_{ij}(p'_{ij})\le \AVGRANK_{ij}(p_{ij},I_{ij})\le r_{ij}/2$.  Each such index can be found in $O(d^{1-\delta})$ time by a brute-force search.  The total time is
$O(d^{1-\delta}n^2)$.
\end{proof}

\subsection{Reducing the Counting Problem to Selected-Entries Matrix Multiplication}

We now describe how to solve our counting problem by using Lopsided \#AE Sparse Triangle.
This part will be similar to previous reductions by Chan, Vassilevska Williams,
and Xu~\cite{ChanWX22}.
 
\begin{lemma}\label{lem:count}
Problem~\ref{prob:count} can be solved in $\Tcount(n,d,\delta)=\OO(d^{1+\delta}\Mselect{n,d^{4+2\delta},n}{n^2/d^{1+\delta}} + d^{1-\delta}n^2)$ time.
\end{lemma}
\begin{proof}
Fix an index $p\in [d]$ and a dyadic interval $I$ of length at least $d^{1-\delta}$.
Let $W$ be a batch of up to $n^2/d^{1+\delta}$ pairs from
$\{(i,j)\in [n]^2: p_{ij}=p,\ I_{ij}=I\}$.

\newcommand{\aalpha}{\widetilde{\alpha}}
\newcommand{\bbeta}{\widetilde{\beta}}
\newcommand{\ccc}{\widetilde{c}}

We concentrate on computing $c_{ij} := \SUMRANK_{ij}(p_{ij},I_{ij})$ for all $(i,j)\in W$.
Computing $\SIZE_{ij}(p_{ij},I_{ij})$ is similar (and easier).

Sort $L=
\{A[i,k]-A[i,k'] : i\in [n],\ k,k'\in [d]\}\cup
\{B[k,j]-B[k',j]: j\in [n],\ k,k' \in[d]\}$, which has $O(d^2n)$ elements,
in $O(d^2n\log n)$ time.
Let $\alpha_{i,k,k'}$ denote the rank of $A[i,k]-A[i,k']$ and 
$\beta_{k,k',j}$ denote the rank of $B[k,j]-B[k',j]$ in the sorted list $L$.
Note that 
\[ c_{ij}\ =\ |\{(k,k')\in I\times [d]: 
\alpha_{i,k',k}\le \beta_{k,k',j}\ \mbox{and}\ \alpha_{i,k,p} < \beta_{p,k,j}
\}|
\]
by the definition of $\SUMRANK_{ij}(p_{ij},I_{ij})$ and Fredman's trick \cite{Fredman76,ChanWX22,ChanWX23} (i.e., the obvious equivalence
of $a+b\le a'+b'$ with $a-a'\le b'-b$).

Define a tripartite graph $G$:
\begin{itemize}
\item The left vertex set is $[n]$, the right vertex set is $[n]$,
and the middle vertex set consists of tuples $(k,k',I_1,I_2)$ where $k\in I$, $k'\in [d]$, and 
$I_1$ and $I_2$ are dyadic intervals in $[O(d^2n)]$ (we do not list all middle vertices exhaustively but generate only those incident to edges).
\item There is an edge between the left vertex $i\in [n]$
and the middle vertex $(k,k',I_1,I_2)$ iff $\alpha_{i,k',k}$ is
in the left half of $I_1$ and $\alpha_{i,k,p}$ is in the left half of $I_2$.
\item There is an edge between the right vertex $j\in [n]$
and the middle vertex $(k,k',I_1,I_2)$ iff $\beta_{k,k',j}+1$ is
in the right half of $I_1$ and $\beta_{p,k,j}$ is in the right half of $I_2$.
\item There is an edge between the left vertex $i\in [n]$ and the
right vertex $j\in [n]$ iff $(i,j)\in W$.
\end{itemize}
Since any point belongs to at most $O(\log n)$ dyadic intervals,
there are a total of $O(d^2n\log^2 n)$ edges between the left part and
the middle part, and between the right part and the middle part.
Let $V_H$ denote the middle vertices of degree more than $\Delta := n/d^{2+2\delta}$
is $\OO(d^2n/\Delta)=\OO(d^{4+2\delta})$.
Let $c^H_{ij}$ be the number of triangles containing $i$ and $j$
through vertices in $V_H$.
Let $c^L_{ij}$ be the number of triangles between $i$ and $j$
not through vertices in $V_H$.
We compute $c^H_{ij}$ for all $(i,j)\in W$ in $O(\Mselect{n,d^{4+2\delta},n}{n^2/d^{1+\delta}})$
time.  We compute $c^L_{ij}$ for all $(i,j)\in W$ by brute force
in $O((d^2n\log^2 (dn))\cdot \Delta) = \OO(n^2/d^{2\delta})$ time.
Then $c_{ij}=c^L_{ij}+c^H_{ij}$.

We repeat the above for all $d$ choices of $p$ and $O(d^\delta)$ choices of $I$, and
a total of $O(d^{1+\delta})$ batches for~$W$.
The total number of iterations is $O(d^{1+\delta})$.
%
%
%
%
%
\end{proof}

Combining all the preceding lemmas, we conclude:
\begin{theorem}
For any constant $\delta>0$,
Real APSP can be solved in $\OO(d^{\delta}n \Mselect{n,d^{4+2\delta},n}{n^2/d^{1+\delta}} + n^3/d^\delta)$ deterministic time.
\end{theorem}

For example, using the fact~\cite[Table 2 in p.~40]{AlmanV26} that
$\Mselect{n,d,n}{w} \le 
\OO(n^2/d^{0.0112} + wd^{0.10})$ for $d\le n^{0.056}$,
the above bound becomes
$\OO(d^\delta n\cdot (n^2/d^{(4+2\delta)0.0112} + (n^2/d^{1+\delta})d^{(4+2\delta)0.10})
+ n^3/d^\delta)$ for $d\le n^{0.056/(4+2\delta)}$.
Setting $\delta={0.0226537}$ and $d=n^{0.013843}$
yields the following (we did not try to optimize the exponent):

\begin{corollary}
Real APSP can be solved in $O(n^{2.9997})$ deterministic time.
\end{corollary}


\section{Real 3SUM Derandomized}\label{sec:real:3sum}

\newcommand{\length}{\mbox{\rm diff-rank}}
\newcommand{\rank}{\RANK}
\newcommand{\SUMP}{\mbox{\rm sum-}\Phi}
\newcommand{\AVGP}{\mbox{\rm avg-}\Phi}
\newcommand{\MIDRANK}{\mbox{\rm mid-rank}}
\newcommand{\SUM}{\mbox{\rm sum-rank-power}}

In this section, we present a truly subquadratic-time deterministic algorithm for Real 3SUM\@. We use ideas similar to those in the previous section, in combination with
Chan, Vassilevska Williams, and Xu's previous randomized reduction~\cite{ChanWX22},
though the details become a bit more involved for 3SUM.

Given sets $A,B,C$ of $O(n)$ elements, we solve the following slight extension of the
3SUM problem: for each $c\in C$, find the predecessor and successor\footnote{
Here, the predecessor of $c$ in a set $S$ is the largest element in $S$ strictly less than than $c$, and the successor is the smallest element in $S$ greater than or equal to $c$.
} of $c$ in 
the sumset $A+B$.  
To this end, we first 
sort $A$ and $B$ and divide the sorted lists of $A$ and $B$ into sublists $A_1,\ldots,A_{n/d}$ and $B_1,\ldots,B_{n/d}$ of size $d$, for a given parameter $d\le n$.
Let $A[i,k]$ denote the $k$-th element of $A_i$ for each $i\in [n/d]$ and $k\in [d]$, and $B[j,\ell]$ denote the $\ell$-th element of $B_j$ for each $j\in [n/d]$ and $\ell\in [d]$.  
By ideas from known slightly subquadratic real 3SUM algorithms (originally due
to Gr\o nlund and Pettie~\cite{GronlundP18}, and continued by Chan~\cite{Chan20}), 
Chan, Vassilevska Williams, and Xu~\cite{ChanWX22} observed that
the original problem
reduces to the following problem
(roughly,  because for each $c\in C$, it suffices to search for $c$ in $A_i+B_j$ for $O(n/d)$ pairs $(i,j)$ that form a monotone ``staircase'' in the $[n/d]\times [n/d]$ grid):

\newcommand{\ThreeSUMOne}{\mbox{\sf Real-3SUM Variant}} %
\begin{problem}[\ThreeSUMOne]\label{prob:3sum1}
We are given an $(n/d)\times d$ real matrix $A$ and an
$(n/d)\times d$ real matrix $B$.
For each $i,j\in [n/d]$, we are also given a set $C_{ij}$ of real
numbers with $\sum_{i,j}|C_{ij}|=O( n^2/d)$.

For each $i,j\in [n/d]$ and each $c\in C_{ij}$, we want to find the predecessor and successor of the value $c$ among the elements in
$\{A[i,k]+B[j,\ell]: k,\ell\in [d]\}$.
\end{problem}

For simplicity of presentation, we again assume that there are no degeneracies (no two pairs have equal sum).

\subsection{A New Counting Problem}

Similar to the halving problem for Real APSP, we define an approximate halving problem for Real 3SUM, 
which shrinks the search space by some constant factor:

\begin{problem}[\textsf{Approximate Halving}]\label{prob:shrink}
We are given an $(n/d) \times d$ real matrix $A$ and an $(n/d)\times d$
real matrix $B$. For each $i, j\in [n/d]$, we are also given a set $Q_{ij}$
of quintuples in $\R\times [d]^4$ with $\sum_{i,j}|Q_{ij}|=O( n^2/d)$.

For each $i, j \in [n/d]$ and each $(c,p^-,q^-,p^+,q^+)\in Q_{ij}$ 
satisfying $A[i,p^-]+B[j,q^-] < c \le A[i,p^+]+B[j,q^+]$,
we want to find a new quadruple
$(p'^-,q'^-,p'^+,q'^+)\in[d]^4$
satisfying $A[i,p'^-]+B[j,q'^+] < c \le  A[i,p'^+]+B[j,q'^+]$,
such that 
\[
\length_{ij}(p'^-_,q'^-,p'^+,q'^+)\le \max\{0.8\cdot \length_{ij}(p^-,q^-,p^+,q^+),\,1\},
\]
where $\length_{ij}(p^-,q^-,p^+,q^+)=\rank_{ij}(p^+,q^+)-\rank_{ij}(p^-,q^-)$
and $\rank_{ij}(p,q)$ is the rank of $A[i,p]+B[j,q]$ 
in $\{A[i,k]+B[j,\ell]: j,\ell\in [d]\}$.
\end{problem}

By iterating $O(\log d)$ times, we can immediately reduce Problem~\ref{prob:3sum1} to the approximate halving problem.

\begin{lemma}
    If Problem~\ref{prob:shrink} can be solved in $\Thalf(n,d)$ time,
    then Problem~\ref{prob:3sum1} can be solved in $O(\Thalf(n,d)\log d)$ time.
\end{lemma}

To solve Problem~\ref{prob:shrink}, we want to find a  good ``pivot'' $(k,\ell)$ 
for which both $\length_{ij}(p'^-_,q'^-,k,\ell)$ and 
$\length_{ij}(k,\ell,p'^+,q'^+)$ are small, i.e., its rank is close to
\[ \MIDRANK_{ij}(p^-,q^-,p^+,q^+)\ :=\ (\RANK_{ij}(p^-,q^-)+\RANK_{ij}(p^+,q^+))/2,
\]
i.e., the following ``potential'' value is small:
\[ \Phi_{ij}(k,\ell,p^-,q^-,p^+,q^+)\ :=\ (\RANK_{ij}(k,\ell) - \MIDRANK_{ij}(p^-,q^-,p^+,q^+))^2.
\]
Inspired by the earlier average rank idea, we will maintain an interval with low average potential value instead, which motivates the following new version of the counting problem:

\begin{problem}[\textsf{Counting}]\label{prob:product}
      We are given an $(n/d) \times d$ real matrix $A$ and an $(n/d)\times d$
real matrix $B$. 
For each $i, j\in [n/d]$, we are also given a set $Q_{ij}$ of quintuples, where each quintuple
is of the form $(p^-,q^-,p^+,q^+,I)$ for some $p^-,q^-,p^+,q^+\in [d]$ and 
$I$ is a dyadic interval in $[d]$ of length at least $d^{1-\delta}$ for a given parameter $\delta>0$, with
$\sum_{i,j}|Q_{ij}|=O(n^2/d)$.  
Let
\[ L_{ij}(p^-,q^-,p^+,q^+,I) := \{(k,\ell)\in I\times [d]: A[i,p^-]+B[j,q^-] < A[i,k]+B[j,\ell]
< A[i,p^+]+B[j,q^+]\}.
\]

For each $i, j \in [n/d]$ and each $(p^-,q^-,p^+,q^+,I)\in Q_{ij}$,
we want to compute:
\begin{eqnarray*}
\SIZE_{ij}(p^-,q^-,p^+,q^+,I) &:=& |L_{ij}(p^-,q^-,p^+,q^+,I)|\\
\SUMP_{ij}(p^-,q^-,p^+,q^+,I) &:=& \sum_{(k,\ell)\in L_{ij}(p^-,q^-,p^+,q^+,I)}
\Phi_{ij}(k,\ell,p^-,q^-,p^+,q^+).
\end{eqnarray*}
\end{problem}


From these values, we can then get:
\[ \AVGP_{ij}(p^-,q^-,p^+,q^+,I)\ :=\ \SUMP_{ij}(p^-,q^-,p^+,q^+,I)/\SIZE_{ij}(p^-,q^-,p^+,q^+,I).
\]

\begin{lemma}
For any constant $\delta>0$, if Problem~\ref{prob:product} can be solved in $\Tcount(n,d,\delta)$ time,
    then Problem~\ref{prob:shrink} can be solved in $\Thalf(n,d)=O(\Tcount(n,d,\delta)\log d+n^2/d^\delta)$ time.    
\end{lemma}
\begin{proof}
Consider a quintuple $(c,p^-,q^-,p^+,q^+)\in Q_{ij}$.
Let $a=\RANK_{ij}(p^-,q^-)$ and $b=\RANK_{ij}(p^+,q^+)$, and 
$\lambda=\length_{ij}(p^-,q^-,p^+,q^+)=b-a$.
Observe that $\SIZE_{ij}(p^-,q^-,p^+,q^+,[d]) = \lambda-1$,
and $\SUMP_{ij}(p^-,q^-,p^+,q^+,[d]) = \sum_{t=a+1}^{b-1} (t - \frac{a+b}{2})^2 =
\lambda(\lambda-1)(\lambda-2)/12$ by straightforward calculations,
and so $\AVGP_{ij}(p^-,q^-,p^+,q^+,[d]) \le \lambda^2/12$, assuming $\lambda>1$.

For the quintuple $(c,p^-,q^-,p^+,q^+)\in Q_{ij}$,
we maintain a dyadic interval $I$ such that
$\SIZE_{ij}(p^-,q^-,p^+,q^+,I)\ge 1$ and 
$\AVGP_{ij}(p^-,q^-,p^+,q^+,I) \le\lambda^2/12$.  At the beginning, we
can set $I=[d]$.

In each round, we divide the dyadic interval $I$ into its left half $I^-$ and  its right half $I^+$.  Similar to before,
$\SIZE_{ij}(p^-,q^-,p^+,q^+,I^-)\ge 1$ and $\AVGP_{ij}(p^-,q^-,p^+,q^+,I^-)\le \lambda^2/12$, or $\SIZE_{ij}(p^-,q^-,p^+,q^+,I^-)\ge 1$ and
 $\AVGP_{ij}(p^-,q^-,p^+,q^+,I^+)\le \lambda^2/12$.
We can tell which is the case by invoking an algorithm for Problem~\ref{prob:count}.
We then reset $I$ to $I^-$ in the former case, or $I^+$ in the latter.

The total number of rounds is $O(\log d)$.
When the length of $I$ drops below $d^{1-\delta}$ for a quintuple $(c,p^-,q^-,p^+,q^+)\in Q_{ij}$,
we find the predecessor $A[i,p'^-]+B[j,q'^-]$ and successor $A[i,p'^+]+B[j,q'^+]$ of $c$ in
 $\{A[i,k]+B[j,\ell]: (k,\ell)\in (I\times [d])\cup\{(p^-,q^-),(p^+,q^+)\}\}$, and return $(p'^-,q'^-,p'^+,q'^+)$
for~$c$.  This can be done by a binary search over $\ell$ for each $k\in I$,
in time $\OO(d^{1-\delta})$ per quintuple in $Q_{ij}$.
The total time over a round is $\OO(n^2/d \cdot d^{1-\delta})=\OO(n^2/d^\delta)$.
To show correctness, note that there exists $(k,\ell)\in L_{ij}(p^-,q^-,p^+,q^+,I)$
with $(\RANK_{ij}(k,\ell)- \frac{a+b}{2})^2\le \AVGP_{ij}(p^-,q^-,p^+,q^+,I) \le  \lambda^2/12$, implying that $\length_{ij}(p^-,q^-,k,\ell)$
and $\length_{ij}(k,\ell,p^+,q^+)$ are both at most $\lambda/2+\lambda/\sqrt{12}<0.8\lambda$.  It follows that $\length_{ij}(p'^-,q'^-,p'^+,q'^+)<0.8\lambda$.
\end{proof}

The above counting problem can be simplified:

\begin{problem}[\textsf{Counting, Simplified}]\label{prob:product2}
Let $s\in\{0,1,2\}$ and $\prec$ be either $\le$ or $<$.
      We are given an $(n/d) \times d$ real matrix $A$ and an $(n/d)\times d$
real matrix $B$. 
For each $i, j\in [n/d]$, we are also given a set $Q_{ij}$ of triples, where each triple
is of the form $(p,q,I)$ for some $p,q\in [d]$ and 
$I$ is a dyadic interval in $[d]$ of length at least $d^{1-\delta}$ for a given parameter $\delta>0$, with
$\sum_{i,j}|Q_{ij}|=O(n^2/d)$.  
Let
\[ L^{\prec}_{ij}(p,q,I) := \{(k,\ell)\in I\times [d]: A[i,k]+B[j,\ell]
\prec A[i,p]+B[j,q]\}.
\]

For each $i, j \in [n/d]$ and each $(p,q,I)\in Q_{ij}$,
we want to compute:
\begin{eqnarray*}
\SUM^{(s,\prec)}_{ij}(p,q,I) &:=& \sum_{(k,\ell)\in L_{ij}(p,q,I)}\RANK_{ij}(k,\ell)^s.
\end{eqnarray*}
\end{problem}

\begin{lemma}
If Problem~\ref{prob:product2} can be solved in $\Tcount'(n,d,\delta)$ time,
then Problem~\ref{prob:product} can be solved in $\Tcount(n,d,\delta)=O(\Tcount'(n,d,\delta))$ time.
\end{lemma}
\begin{proof}
Computing $\RANK_{ij}(p^-,q^-)$ and $\RANK_{ij}(p^+,q^+)$, 
and thus $\MIDRANK_{ij}(p^-,q^-,p^+,q^+)$, reduces to
computing $\SUM_{ij}^{(0,\le)}(p^-,q^-,[d])$ and $\SUM_{ij}^{(0,\le)}(p^+,q^+,[d])$.

Computing $\SIZE_{ij}(p^-,q^-,p^+,q^+,I)$ reduces to 
computing $\SUM_{ij}^{(0,<)}(p^+,q^+,I)$ and $\SUM_{ij}^{(0,\le)}(p^-,q^-,I)$
by subtraction.

Given $m=\MIDRANK_{ij}(p^-,q^-,p^+,q^+)$,
computing $\SUMP_{ij}(p^-,q^-,p^+,q^+,I)$ reduces to
computing $\SUM^{(s,<)}(p^+,q^+,I)$ and $\SUM^{(s,\le)}(p^-,q^-,I)$ 
by subtraction, for all $s\in\{0,1,2\}$, since $(x-m)^2$ expands to $x^2-2mx +m^2$.
\end{proof}

\subsection{Reducing the Counting Problem to Selected-Entries Matrix Multiplication}

We now solve this simplified counting problem using Lopsided \#AE Sparse Triangle,
by adapting previous reductions by Chan, Vassilevska Williams,
and Xu~\cite{ChanWX22}.
 
\begin{lemma}\label{lem:count}
Problem~\ref{prob:product2} can be solved in $\Tcount'(n,d,\delta)=\OO(d^\delta \Mselect{n,d^{12+2\delta},n}{n/d^{1+\delta}} + n^2/d^\delta)$ time.
\end{lemma}
\begin{proof}
Fix a dyadic interval $I$ of length at least $d^{1-\delta}$.
Let $W$ be a batch of up to $n^2/d^{1+\delta}$ pairs from
$\{(i,p,j,q)\in [n/d]\times [d]\times[n/d]\times[d]: (p,q,I)\in Q_{ij}\}$.

 We concentrate on the case $s=2$, since the cases $s=0$ and $s=1$ are easier.  W.l.o.g., say $\prec$ is $<$.  Then 
\begin{eqnarray*}
\lefteqn{c_{(i,p),(j,q)}\ :=\ \SUM^{(2,<)}_{ij}(p,q,I)}\\
& =& |\{(k,\ell,k',\ell',k'',\ell'')\in I\times [d]^5:  A[i,k]+B[j,\ell] < A[i,p]+B[j,q]\ \mbox{and}\\
&& \qquad A[i,k']+B[j,\ell']\le A[i,k]+B[j,\ell]\ \mbox{and}\ A[i,k'']+B[j,\ell'']\le A[i,k]+B[j,\ell]\}|.
\end{eqnarray*}

Sort $L=
\{A[i,k]-A[i,k'] : i\in [n/d],\ k,k'\in [d]\}\cup
\{B[j,\ell]-B[j,\ell']: j\in [n/d],\ \ell,\ell' \in[d]\}$, which has $O(dn)$ elements,
in $O(dn\log n)$ time.
Let $\alpha_{i,k,k'}$ denote the rank of $A[i,k]-A[i,k']$ and
$\beta_{j,\ell,\ell'}$ denote the rank of $B[j,\ell]-B[j,\ell']$ in the sorted list $L$.
By Fredman's trick,
\begin{eqnarray*}
c_{(i,p),(j,q)}
&=& |\{(k,\ell,k',\ell',k'',\ell'')\in I\times [d]^5:  \alpha_{i,k,p} < \beta_{j,q,\ell}\ \mbox{and}\ \alpha_{i,k',k} \le \beta_{j,\ell,\ell'}\ \mbox{and}\ \alpha_{i,k'',k} \le \beta_{j,\ell,\ell''}\}|.
\end{eqnarray*}

Define a tripartite graph $G$:
\begin{itemize}
\item The left vertex set is $[n/d]\times[d]$, the right vertex set is $[n/d]\times[d]$,
and the middle vertex set consists of tuples $(k,\ell,k',\ell',k'',\ell'',I_1,I_2,I_3)$ where $(k,\ell,k',\ell',k'',\ell')\in I\times [d]^5$ and
$I_1,I_2,I_3$ are dyadic intervals in $[O(dn)]$ (we do not list all middle vertices exhaustively but generate only those incident to edges).
\item There is an edge between the left vertex $(i,p)$
and the middle vertex $(k,\ell,k',\ell',k'',\ell'',I_1,I_2,I_3)$ iff
$\alpha_{i,k,p}$ is in the left half of $I_1$,
$\alpha_{i,k',k}$ is in the left half of $I_2$,
and $\alpha_{i,k'',k}$ is in the left half of $I_3$.
\item  There is an edge between the right vertex $(j,q)$
and the middle vertex $(k,\ell,k',\ell',k'',\ell'',I_1,I_2,I_3)$ iff
$\beta_{j,q,\ell}$ is in the right half of $I_1$,
$\beta_{j,\ell,\ell'}+1$ is in the right half of $I_2$,
and $\beta_{j,\ell,\ell''}+1$ is in the right half of $I_3$.
\item There is an edge between the left vertex $(i,p)$ and the
right vertex $(j,q)$ iff $(i,p,j,q)\in W$.
\end{itemize}
Since any point belongs to at most $O(\log n)$ dyadic intervals,
there are a total of $O((n/d)\cdot d\cdot  d^6\log^3n)=\OO(d^6n)$ edges between the left part and
the middle part, and between the right part and the middle part.
Let $V_H$ denote the middle vertices of degree more than $\Delta := n/d^{6+2\delta}$
is $\OO(d^6n/\Delta)=\OO(d^{12+2\delta})$.
Let $c^H_{(i,p),(j,q)}$ be the number of triangles containing $(i,p)$ and $(j,q)$
through vertices in $V_H$.
Let $c^L_{(i,p),(j,q)}$ be the number of triangles containing $(i,p)$ and $(j,q)$
not through vertices in $V_H$.
We compute $c^H_{(i,p),(j,q)}$ for all $(i,p,j,q)\in W$ in $O(\Mselect{n,d^{12+2\delta},n}{n^2/d^{1+\delta}})$
time.  We compute $c^L_{(i,p),(j,q)}$ for all $(i,p,j,q)\in W$ by brute force
in $\OO(d^6n\cdot \Delta) = \OO(n^2/d^{2\delta})$ time.
Then $c_{(i,p),(j,q)}=c^L_{(i,p),(j,q)}+c^H_{(i,p),(j,q)}$.

We repeat the above for all $O(d^\delta)$ choices of $I$, and
a total of $O(d^{\delta})$ batches for $W$.
The total number of iterations is $O(d^{\delta})$.
\end{proof}

Combining all the preceding lemmas, we conclude:
\begin{theorem}
For any constant $\delta>0$, Real 3SUM can be solved in $\OO(d^{\delta} \Mselect{n,d^{12+2\delta},n}{n^2/d^{1+\delta}} + n^2/d^\delta)$ deterministic time.
\end{theorem}


For example, using the fact~\cite[Corollary 31 with $c=21$ and $\theta=0.0145$]{AlmanV26} that
$\Mselect{n,d,n}{w} \le 
\OO(n^2/d^{0.008352} + wd^{0.077647})$ for $d\le n^{0.056853}$,
the above bound becomes
$\OO(d^\delta \cdot (n^2/d^{(12+2\delta)0.008352} + (n^2/d^{1+\delta})d^{(12+2\delta)0.077647})
+ n^2/d^\delta)$ for $d\le n^{0.056853/(12+2\delta)}$.
Setting $\delta={0.050534}$ and $d=n^{0.0046983}$
yields the following (we did not try to optimize the exponent):







\begin{corollary}
Real 3SUM can be solved in $O(n^{1.9998})$ deterministic time.
\end{corollary}

The only primitive  operations used on the input real numbers are comparisons of sums of pairs with a third number, and comparisons of sums/differences of two pairs.

We can also obtain a truly subcubic deterministic algorithm for Real Exact Triangle by using the same strategy, in combination with the previous reduction~\cite{ChanWX22} (in fact, Exact Triangle is a little easier to work with than 3SUM).

\subsection*{AI Disclosure}

The results in 
this paper were obtained using ChatGPT Sol 6.1.  
Initial proofs 
were complicated and difficult to digest. 
This manuscript is completely human-written: we (the authors) rewrote the proofs from scratch, simplified many parts, and presented the solutions in our own way.
We have verified the correctness and originality of all content, including references.

\small
\bibliographystyle{alphaurl}
\bibliography{ref}

\end{document}